\documentclass{article}
\usepackage{amssymb,amsmath,amstext,amsfonts, bbm,amsthm}
\newtheorem{theorem}{\sc Theorem}
\newtheorem{lemma}{\sc Lemma}
\newtheorem{prope}{\sc Property}
\newtheorem{coro}{\sc Corollary}
\newtheorem{nota}{\sc Notation}
\newtheorem{defin}{\sc Definition}
\newtheorem{cla}{\sc Claim}
\newtheorem{rem}{\sc Remark}

\newenvironment{corollary}{\begin{coro}}{\end{coro}}

\newenvironment{definition}{\begin{defin}}{\end{defin}}

\newcommand{\ld}[1]{\mbox{depth}_{#1}}

\usepackage{lineno,hyperref}
\modulolinenumbers[5]

\begin{document}

%\begin{frontmatter}

\title{Logical depth for reversible Turing machines with an application to the
rate of decrease in logical depth for general Turing machines} 

\author{Paul M.B. Vit\'anyi\thanks{CWI and University of Amsterdam. Address:
CWI, Science Park 123, 1098XG Amsterdam, The Netherlands.
Email: Paul.Vitanyi@cwi.nl}}

\maketitle
%\end{frontmatter}

\begin{abstract}
The logical depth of a {\em reversible} Turing machine
equals the shortest running time of a shortest program for it. This
is applied to show that the result in 
\cite{ASV17} is valid notwithstanding the error noted in  
Corrigendum~\cite{Vi18}.

keywords: Logical depth,
Kolmogorov complexity,
compression

\end{abstract}
\section{Introduction}
A book on number theory is
difficult, or `deep.' The book lists
a number of difficult theorems of number theory. However,
it has very low Kolmogorov complexity, since all
theorems are derivable from the initial few definitions.
Our estimate of the difficulty, or `depth,' of the book is based
on the fact that it takes a long time to reproduce the book
from part of the information in it.
The existence of a deep book is itself evidence of some long
evolution preceding it.

The logical depth of a (finite) string is related to complexity with bounded resources and measures the tradeoff between program sizes and running times.
Computing a string~$x$ from one of its shortest programs
may take a very long time, but computing the same string from a simple ``\verb print \verb$(x)$''
program of length about $|x|$ bits takes very little time.

Logical depth as defined in \cite{ben88} for a string 
 comes in two versions:
one based on the compressibility of programs of prefix Turing machines
and the other
using the ratio between algorithmic probabilities with and without time limits.
Since both are approximately the same (\cite[Theorem 7.7.1]{LV08} based 
on \cite[Lemma 3]{ben88}) it is no restriction
to use the compressibility version. 

The used notions of computability, resource-bounded computation time, 
self-delimiting strings,
big-O notation, and Kolmogorov complexity are
well-known and the properties, notations, are treated in \cite{LV08}.

\section{Preliminaries}
All Turing machines in this paper are prefix Turing machines.
A {\em prefix Turing machine} 
is a Turing machine with a one-way read-only
program tape, an auxiliary tape, one or more work tapes
and an output tape. All tapes are linear one-way infinite 
and divided into cells
capable of containing
one symbol out of a finite set. 
Initially the program tape is inscribed with an
infinite sequence of 0's and 1's and the head is scanning the leftmost cell.
When the computation terminates
the sequence of bits scanned on the input tape is the {\em program}.
For every fixed finite contents of the auxiliary tape
the set of programs for such a machine
is a prefix code (no program is a proper prefix of another program). 
%Another
%consequence is that if the computation of a prefix Turing machine takes
%$d$ steps then for its program $p$ holds that $d \geq |p|$.
Let $T_0,T_1, \ldots$ be the standard enumeration of 
prefix Turing machines.
A {\em universal prefix Turing machine} simulates every prefix 
Turing machine given its index number. We also require it to be {\em optimal}
which means that the simulation program is as short as possible.
We choose a {\em reference
optimal universal prefix Turing machine} and call it~$U$.

The prefix Kolmogorov complexity is based on the prefix Turing machine similar
to the (plain) Kolmogorov complexity based on the (plain) Turing machine.
%A self-delimited version of a binary string $x$ is 
%$x'= 1^{||x||}0 |x|x$ where
%$||x||$ denotes the length of $|x|$, and $|x|$ denotes the length of $x$.
%The length of this self-delimiting version $x'$ of 
%$x$ with $|x|=n$ is 
%$n+2 \log n+1$  where $\log$ denotes the logarithm base 2.
Let $x,y$ be finite binary strings.
The {\em prefix Kolmogorov complexity}
$K(x|y)$ of $x$ with auxiliary $y$ is defined by
\[
  K(x|y)=\min_p \{|p| : U(p,y)= x \}.
\]
If $x$ is a binary string of length $n$ then
$K(x|y) \leq n+O( \log n)$.
Restricting the computation time resource is
indicated by a superscript giving the allowed number of steps, usually
denoted by $d$. 
The notation $U^d(p,y)=x$ means that $U(p,y)=x$ within
$d$ steps.
%The {\em $d$-time-bounded
%prefix Kolmogorov complexity} $K^{d}(x|y)$ is defined by
%\[
%K^{d}(x|y) =\min_p \{|p| : U^d(p,y)=x \}.
%\]
If the auxiliary string $y$
is the empty string $\epsilon$, then we usually
drop it. Similarly, we write $U(p)$ for $U(p,\epsilon)$.
The string $x^*$ is a {\em shortest program} for $x$ if $U(x^*)=x$ and
$K(x)=|x^*|$.
A string $x$ is $b$-{\em incompressible} if $|x^*| \geq |x|-b$. 

%\begin{definition}
%\rm
%Let $x$ be a string.
%The {\em logical depth} of $x$ is
%\[
%\min \left\{d: p\in \{0,1\}^* \wedge U^d(p) = x \wedge |p| \leq K(p)\right\},
%\]
%the least number of steps to compute $x$ from an incompressible 
%program for $U$.
%\end{definition}

\section{Reversible Turing Machines}

A Turing machine behaves according to a finite
list of rules.
These rules determine, from the current state of the finite
control and the symbol contained in the cell under scan,
the operation to be performed next and the state to enter
at the end of the next operation execution.
%The rules have format $(p,s,a,q)$ with $p \in Q$  ($Q$ finite) is the
%current state of the finite control; $s \in \{s_1, \ldots ,s_n\}=S$ 
%is the symbol
%under scan; $a$ is the next operation to be executed
%of type `move one cell left,' `move one cell right,' or `print $s \in S$,'  
%designated in the obvious sense by an element from
%$A =  S \bigcup \{L,R \}$; and $q \in Q$ 
%is the state of the finite control
%to be entered at the end of this step.

%We require that distinct quadruples cannot have their first
%two elements identical: 
The device is
({\em forward}) {\em deterministic}.
Not every possible combination of the first two elements has to be
in the set; in this way we permit the device
to perform {\it no} operation. In this case we say that the device
{\em \it halts}. 
Hence, we can define a Turing machine
by a {\em transition function}. % $\pi$ from a finite subset of $Q \times S$
%into $A \times Q$.

%To define a reversible Turing machine we restrict the transtions as follows.
%A transition is either {\em moving}, that is $a \in {R,L}$, or 
%{\em stationary}, that is $a \not\in {R,L}$.
%Each moving transition is oblivious (if $p,s,a,q)$ is the transition then
%$p,s',a,q)$ for all $s' \in S$ are transitions. Moreover,
%if $p,s,a,q)$ and ????

\begin{definition}
\rm
A {\em reversible} Turing machine \cite{Be73,AG11} is a 
Turing machine that
is forward deterministic (any Turing mchine as defined is) 
but also {\em backward deterministic}, that is, 
the transition function has a single-valued inverse.
The details of the formal definition
are intricate \cite{Be73,AG11} and need not concern us here.
This definition extends in the obvious manner to multitape Turing
machines.
\end{definition}
In \cite{Be73} for every  1-tape ordinary Turing machine $T$ 
a 3-tape reversible Turing machine $T_{\rm rev}$ 
is constructed that emulates $T$ in linear 
time such that with input $p$ the output is 
$T_{\rm rev}(p)= (p,T(p))$. 
The reversible Turing machine that emulates $U$ is 
called $U_{\rm rev}$. 
\begin{definition}
\rm
Let $x$ be a string and $b$ a nonnegative integer.
The {\em logical depth} of $x$ {\em at significance level}
$b$, is
\[
\ld{b}(x) = \min \left\{d: p\in \{0,1\}^* \wedge U^d(p) = x \wedge |p| \leq K(p) + b \right\},
\]
the least number of steps to compute $x$ by a $b$-incompressible program.
\end{definition}

\begin{theorem}\label{theo.1}
The logical depth of a string $x$ at significance level $b \in {\cal N}$ 
for reversible Turing machines is equal to 
\[
\ld{b}(x) = \min \left\{d: p \in \{0,1\}^*  \wedge U^d_{\rm rev}(p) = (p,x) \wedge |p| \leq K(x) + b \right\},
\]
the least number of steps to compute $x$ by $U_{\rm rev}$ 
from a program of length at most $K(x)+b$.
\end{theorem}
\begin{proof}
Since a reversible Turing machine is backwards deterministic, 
and an incompressible program cannot be computed from a shorter program,
the length of an incompressible program for $x$ can only be the
length of a shortest program for $x$. The logical depth at significance $b$
is then the least number of steps to compute $x$ by $U_{\rm rev}$
from a program $p$ of length $K(x)+b$. 
\end{proof}

\section{The Rate of Decrease of Logical Depth}

In \cite[Section 4 ]{ASV17}
it is assumed that, for all $x \in \{0,1\}^*$, the string $x^*$ is the 
only incompressible string
such that $U(x^*)=x$. That is, logical depth according to 
Theorem~\ref{theo.1} is used.
However, this assumption is wrong for general Turing machines  in
that for many $x$ there may be an incompressible
string $p$ with $|x| \geq |p| > |x^*|$ such that $U(p)=x$. The computation 
of $U(p)=x$ may be faster than that of $U(x^*)= x$. 
For example, the function from $x \in \{0,1\}^*$ to the least number of steps
in a computation 
$U(p)=x$ for an incompressible string $p$ may be computable. 
The argument in the
paper is, however, correct for the set of reversible Turing machines. 
These Turing machines are a subset of the set of all 
Turing machines \cite{Be73,AG11} and emulate them in linear time. 
This implies the correctness of \cite[Theorem 2]{ASV17} as we shall show.

\begin{lemma}
Let $\psi$ be defined by
\[
\psi(n) = \max_{|x|=n}\min_d\{d:U^d_{\rm rev}(x^*)=(x^*,x)\}.
\]
Then $\psi$ is not computable and grows faster than any computable function.
\end{lemma}
\begin{proof}
If a function $\psi$ as in the lemma were computable, then for an $x$
of length~$n$ we could run $U_{\rm rev}$ emulating $U$ \cite{Be73} 
forward for $\psi (n)$ 
steps on all programs of length $n+O(\log n)$. Among those
programs that halt within $\psi(n)$ steps, we could select 
the programs $p$ which output $(p,x)$.
Subsequently, we could select from that set a program $p$ of minimum length,
say $x^*$.
%from which $(x^*,x)$ can be computed by the reversible 
%Turing machine $U_{\rm rev}$ emulating $U$ \cite{Be73}. 
Such a program $x^*$ has length $K(x)$ since $U_{\rm rev}$ is  emulating $U$.
This would imply that $K$ would be computable.
But the function $K$ is incomputable \cite{Ko65,LV08}: contradiction. 
Therefore $\psi$
cannot be computable.
Since this holds for every function majoring $\psi$,
the function $\psi$ must grow faster than any computable function.
\end{proof}

\begin{corollary}
\rm
The set of reversible Turing machines is a subset of the set of 
all Turing machines. The emulation of $U(p)$ by $U_{\rm rev} (p)$ is 
linear time for all binary inputs $p$ by \cite{Be73}. 
Therefore, replacing in the lemma $U_{\rm rev}$ by $U$ changes $\psi(n)$ to 
$\phi(n) = \Omega ( \psi(n))$. Hence  the lemma holds with $\psi$
replaced by $\phi$ and $U_{\rm rev}$ by $U$. This gives us
\cite[Lemma 1]{ASV17} and therefore \cite[Theorem 2]{ASV17}
(the Busy Beaver upper bound is proved as it is in \cite{ASV17}):
\begin{theorem}
The function
\[
f(n)= \max_{|x|=n, \; 0\leq b \leq n} 
\{x:\ld{b}(x) - \ld{b+1}(x)\}
\]
grows faster
than any computable function but not as fast as the Busy Beaver function.
\end{theorem}
 
\end{corollary}

{\small
\bibliographystyle{plain}

\end{document}